%% file: main.tex
\documentclass{article}
\usepackage[preprint]{neurips_2026}

\usepackage[utf8]{inputenc}
\usepackage{booktabs}
\usepackage{tabularx}

\input{commands}

\title{Single-Item Auctions with a Monopolist Intermediary}

\author{Jingyi Liu\thanks{Princeton University} \And
        Aviad Rubinstein\thanks{Stanford University} \And
        Ertem Nusret Tas\thanks{a16z Crypto Research} \And
        S. Matthew Weinberg\thanks{Princeton University} \And
        Qianfan Zhang\thanks{Princeton University}}

\begin{document}

\begin{titlepage}
\clearpage\maketitle
\thispagestyle{empty}
\begin{abstract}

Classical optimal auction theory assumes that bids reach the seller directly. We study how this picture changes when a revenue-maximizing intermediary controls access to the seller’s auction. Motivated by blockchain auctions, online platforms, and other intermediated markets, we consider a single-item auction with independent private values and a monopolist intermediary who can decide which bidder messages are forwarded to the seller. We establish approximation guarantees and impossibility results across three timing models: seller-first, intermediary-first, and simultaneous. In the seller-first model, arbitrary deterministic seller mechanisms collapse to posted-price mechanisms, and the intermediary’s best response is a shifted Myerson auction. This yields a sharp separation: for regular distributions, the seller’s revenue can be arbitrarily small relative to the no-intermediary optimum, while for $\alpha$-strongly regular distributions, posted prices recover a constant fraction of the optimum with a tight dependence on $\alpha$. We further show that timing matters: neither Stackelberg order uniformly dominates, and simultaneous play can leave both parties unboundedly worse off than in either sequential model.

\end{abstract}

\end{titlepage}

\input{introduction}

\input{model}

\input{honest_BP}
\input{seller_first}

\input{intermediary_first}

\input{simultaneous}

\section*{Acknowledgements}
The work of ENT was conducted mostly while visiting Princeton University and at Stanford University.

\bibliographystyle{alpha}
\bibliography{references}

\appendix

\input{appendix_related_work}

\input{appendix}

\end{document}

%% file: commands.tex
\usepackage{graphicx}
\usepackage{xspace}
\usepackage{amsthm}
\usepackage{amsmath}
\usepackage{amssymb}
\usepackage{tabularx}

\usepackage{hyperref}
\usepackage{tikz}
\usetikzlibrary{arrows.meta}
\tikzset{>={Latex[length=4,width=4]}}
\usetikzlibrary{calc,intersections,decorations.markings}
\usepackage{siunitx}
\usepackage{xcolor}

\usepackage[capitalise]{cleveref}

\usepackage{xspace}

\newcommand{\ie}[0]{\emph{i.e.}\xspace}
\newcommand{\eg}[0]{\emph{e.g.}\xspace}
\newcommand{\cf}[0]{\emph{cf.}\xspace}

\renewcommand{\Pr}{\operatorname*{\mathbf{Pr}}}
\DeclareMathOperator*{\E}{\mathbf{E}}
\DeclareMathOperator*{\I}{\mathbf{1}}

\newtheorem{theorem}{Theorem}[section]
\newtheorem{definition}[theorem]{Definition}
\newtheorem{lemma}[theorem]{Lemma}
\newtheorem{corollary}[theorem]{Corollary}

\newtheorem{remark}[theorem]{Remark}
\newtheorem{assumption}[theorem]{Assumption}

\newcommand{\OPT}[0]{\ensuremath{\mathsf{OptRev}}}
\newcommand{\AP}[0]{\ensuremath{\mathsf{AP}}}
\newcommand{\APRev}[0]{\ensuremath{\mathsf{APRev}}}

\newcommand{\SRev}[0]{\ensuremath{\mathsf{SelRev}}}
\newcommand{\IRev}[0]{\ensuremath{\mathsf{IntRev}}}
\newcommand{\Rev}[0]{\ensuremath{\mathsf{Rev}}}
\newcommand{\BR}[0]{\ensuremath{\mathsf{BR}}}
\newcommand{\OptRev}[0]{\ensuremath{\mathsf{OptRev}}}

\newcommand{\argmax}[0]{\ensuremath{\mathrm{argmax}}}

\newcommand{\vb}{\boldsymbol}
\renewcommand{\epsilon}{\varepsilon}

\newcommand{\mycomment}[1]{}

%% file: introduction.tex
\section{Introduction}
\label{sec:introduction}

In classical auction theory, bidders submit messages, the mechanism observes them, and the allocation rule is applied.
However, in modern digital platforms, bids and orders often pass through an economically interested routing layer.
Recently, the emergence of open-source payment and connection standards for AI agents (\eg, x402~\cite{x402_whitepaper}, MCP~\cite{mcp}) enabled these agents to engage in e-commerce, an increasingly crucial part of their interactions on the web~\cite{doorenbos1997scalable,greenwald1999shopbots,zhou2024webarena}.
Agent-server interactions often happen over intermediary platforms that curate access to online data~\cite{merit_systems,merit_blog}, or AI models~\cite{hugging_face,replicate,open_router}.
As these platforms gain traction, they would have more power over which bids are visible to which sellers, mirroring a similar trend in online advertising (see~\Cref{sec:related-work-main} for works motivated by ad auctions).
This prediction of monopoly power over forwarded bids calls for an analysis of the \emph{principal-agent problem}, where the platform extracts revenue from the auction by censoring other bids and inserting its own bids.

In this work, we consider single-item auctions with a single monopolist intermediary that models the economically-interested web platforms.
Reflecting the monopoly power of these platforms, the intermediary runs an auction to sell the \emph{exclusive} right to dictate the bids forwarded to the seller.
Both the seller and the intermediary aim to optimize their revenues given the other party.
In this context, we consider three different models: (i) the seller-first model, (ii) the intermediary-first model, and (iii) the simultaneous model. 
In the seller-first model, we analyze a Stackelberg game, where the seller commits to its mechanism before the intermediary, and the intermediary observes the seller's mechanism before committing to its own.
In the intermediary-first model, the roles are swapped, and the intermediary leads the game.
Finally, in the simultaneous model, we analyze the Nash equilibria that arise when the seller and the intermediary reveal their mechanisms simultaneously.

There is a rich literature on auctions with multiple intermediaries in the seller-first model (\cf~Appendix~\ref{sec:appendix_related-work}).
Although a single intermediary with captive bidders is a special case, our work provides explicit tighter bounds on the revenues under a wider class of regular distributions, while also analyzing the intermediary-first and simultaneous models.

\subsection{Contributions}
We show that a single strategic intermediary who controls access to the seller’s auction fundamentally changes the structure of optimal auction design. Even in a single-item setting, the intermediary can exploit access control to impose an additional layer of monopoly pricing. We analyze this effect across three timing models. For a summary of concrete approximation ratios achieved in each setting, we refer the readers to Appendix~\ref{appendix:summary-results}.

\subsubsection{Seller-First Model}
In the seller-first model, the seller commits to a mechanism before the intermediary moves. Our first insight is a structural collapse: any deterministic seller mechanism is revenue-equivalent to a posted price mechanism, where the posted price $p_0$ equals the minimum payment needed to obtain the item from the seller in the seller's auction. This collapse occurs because the winner of the intermediary’s auction can dictate the seller-facing bid vector and will choose the bid vector that obtains the item at the cheapest cost. We then characterize the intermediary’s best response. Given the seller's posted price $p_0$, we prove that the intermediary's best response is to run a shifted Myerson's auction with personalized reserves of $\phi^{-1}_i(p_0)-p_0$ for the bidder $i$ 
and to allow the winner to acquire the item at price $p_0$, where $\phi_i(\cdot)$ is the virtual value function for bidder $i$. The effective reserve price for bidder $i$ then becomes $\phi^{-1}_i(p_0)$, which is equal to the combined reserve price of the seller ($p_0$) and the intermediary ($\phi^{-1}_i(p_0)-p_0$).
Note that $\phi^{-1}_i(p_0)$ is larger than the optimal reserve $\phi^{-1}_i(0)$ in the absence of the intermediary.
The gap between the two prices implies a loss of combined revenue in the presence of the intermediary.

This characterization yields sharp revenue consequences. For regular distributions, we exhibit an example with a single bidder such that the seller’s revenue can be arbitrarily small compared to the \emph{optimal revenue} (\ie, the revenue in the absence of the intermediary), which in this case is fully captured by the intermediary. In contrast, for $\alpha$-strongly regular distributions ($0<\alpha\leq 1$), anonymous posted-price mechanisms recover a constant fraction of the optimal revenue, and the dependence on $c(\alpha)=\alpha^{1/(1-\alpha)}$ is tight up to universal constants. Finally, we prove that adopting a randomized mechanism does not help either the seller or the intermediary gain more revenue, enabling us to restrict our attention to deterministic mechanisms.

\subsubsection{Intermediary-First Model}

In the intermediary-first model, the intermediary commits to a mechanism first before observing the seller’s eventual auction, while the seller can best respond to the access probability induced by the intermediary's mechanism. This reverses the strategic advantage present in the seller-first model. For regular distributions, we show that the intermediary's revenue can be arbitrarily small compared to the optimal revenue, even with a single bidder. Conversely, for i.i.d. MHR (1-strongly regular) distributions, a simple posted-price intermediary mechanism guarantees a constant fraction of the optimal revenue for the intermediary, who moves first in this setting. Together, the two Stackelberg models show that commitment power has first-order consequences for revenue division: the party that moves first can often secure a constant-fraction guarantee for broad classes of bidder distributions, but neither timing order uniformly dominates the other.

\subsubsection{Simultaneous Model}
In the simultaneous model, the seller and the intermediary reveal their mechanisms at the same time, and we study the Nash equilibria formed by the seller and the intermediary. We show that the game can have infinitely many Nash equilibria, even for the case of a single bidder with a constant value, which shows that the simultaneous model is highly sensitive to equilibrium selection.
Additionally, we describe $\alpha$-strongly regular value distributions for some $\alpha \in (0,1)$, such that for any tuple of Nash equilibrium strategies used by the seller and the intermediary, both parties' revenues are \emph{unboundedly} worse than their revenues in any Stackelberg game, where the seller or the intermediary moves first.
Hence, both parties can be worse off in the simultaneous model, as compared to the seller or intermediary-first models.

\subsection{Related Work}
\label{sec:related-work-main}
In Appendix~\ref{sec:appendix_related-work}, we provide an extended related work section on auctions with intermediaries, multi-level mechanisms, and the implications of our results for auctions over blockchains.
Although there are many related works on auctions with intermediaries~\cite{FMM+10,SGP12,SGP13,SGP14,BC17,BLM17}, here we highlight one paper that is most relevant to our results (see~Appendix~\ref{sec:related-work-awi} for an extended discussion).
Aggarwal et al. investigate multi-item auctions with multiple intermediaries, each representing a disjoint set of bidders in a model akin to our seller-first model, where the intermediaries can best respond to the seller's mechanism~\cite{ABGP22}. Their main abstraction is a buyer-disconnectability parameter $\tau$: when the seller posts a uniform per-item price $r$, if $q$ buyers represented by an intermediary have values at least $r$, then the intermediary demands at least $\tau q$ items from the seller in expectation. The authors provide an intermediary-proof mechanism for the seller that always guarantees a $\tau c(\alpha) (1-\frac{1}{e})$-approximation to the optimal revenue for $\alpha$-strongly regular distributions where $0<\alpha<1$ and a $\tau \frac{1}{e} (1-\frac{1}{e})$-approximation for MHR distributions.\footnote{We translate the notation of~\cite{ABGP22} into ours as follows.
Their \(\lambda\)-regularity parameter runs in the opposite direction from our
\(\alpha\)-strong regularity parameter: \(\lambda=0\) corresponds to MHR and
\(\lambda=1\) corresponds to regularity, whereas \(\alpha=1\) corresponds to
MHR and \(\alpha=0\) corresponds to regularity. Thus \(\alpha=1-\lambda\), and
their \(c(\lambda)=(1-\lambda)^{1/\lambda}\) becomes
\(c(\alpha)=\alpha^{1/(1-\alpha)}\).} In the case of a revenue-maximizing monopolist intermediary, they show that $\tau\geq c(\alpha)$ for $\alpha$-strongly regular distributions, which translates to a $c(\alpha)^2 (1-\frac{1}{e})$-approximation guarantee for $\alpha$-strongly regular distributions.

Our seller-first model is closely related but takes a different approach: rather than abstracting intermediary behavior through a pass-through parameter, we explicitly model a single monopolist intermediary who controls access to the seller’s auction and characterize its exact best response. This yields a shifted-Myerson characterization and avoids the additional black-box buyer-disconnectability loss, giving a $c(\alpha)$-type dependence, instead of $c(\alpha)^2$. We also handle non-identical bidder distributions, which are not covered by the i.i.d. identical-item guarantee of~\cite{ABGP22}. Finally, for regular distributions, we show that the seller’s revenue can be arbitrarily small relative to the no-intermediary optimum. This gives a negative answer, in the monopolist-intermediary setting, to their open question of whether intermediary-proof seller mechanisms exist for regular bidder valuations.

\subsection{Discussion}
Our work was motivated by the potential effect of intermediaries in online markets used by autonomous agents. Our goal is to improve the understanding of how strategic access-controlling intermediaries can distort auction outcomes in online platforms, blockchain auctions, and other intermediated markets.
Although our characterization of optimal intermediary behavior could potentially help intermediaries extract more revenue from sellers or bidders, our work is intended to level the playing field by identifying how strategic access control can harm sellers and showing that, in the seller-first model where the seller has the commitment power, simple posted-price mechanisms can recover constant-factor revenue guarantees under appropriate distributional assumptions.

\subsection{Outline}
Section~\ref{sec:model} formalizes the model.
The analysis of the seller-first, intermediary-first and simultaneous models are provided in Sections~\ref{sec:seller-first},~\ref{sec:bp-first} and~\ref{sec:simultaneous} respectively.

%% file: model.tex
\section{Model}
\label{sec:model}

There are three parties in our model: a \emph{seller}, a monopolist \emph{intermediary}, and $n$ \emph{bidders}.
The seller runs an auction to sell a single indivisible item to $n$ bidders, while the intermediary has full control over the set of bids that can reach the seller's auction, and it uses a mechanism to ask bidders to pay to be included in the seller's auction, or even pay for the right to dictate the auction.
Both the seller and the intermediary seek to maximize their own revenue.

Each bidder $i \in [n]$ has a private value $v_i \in \mathbb{R}_{\ge 0}$ for the item, which is drawn independently according to a distribution $F_i$.\footnote{We will use $F_i$ to denote both the distribution and also its CDF.} 
All these distributions $F_i$ are public information to both the seller and the intermediary.
Note that, since the intermediary can decide what bids to forward to the seller, the seller might only see a subset of those $n$ bidders.
We will formally define the game between the seller and the intermediary in the following section. We assume that all mechanisms are deterministic unless stated otherwise.

\subsection{The Seller-Intermediary Game}

In the \emph{seller-intermediary game}, the seller must choose a mechanism $\Pi_S=(\vb{x},\vb{p})$ that maps bids $\hat{b}_1,\ldots,\hat{b}_n$ to an allocation $\vb{x}(\hat{b}_1,\ldots,\hat{b}_n) \in \{0,1\}^n$ and payment $\vb{p}(\hat{b}_1,\ldots,\hat{b}_n) \in \mathbb{R}_{\ge 0}^n$ where the bidder associated with $\hat{b}_j$ (who might  not be the $j$-th bidder with distribution $F_j$) such that $x_j=1$ receives the item with a payment $p_j$.
For convenience, denote the ``true'' allocation and payment by $\vb{\hat{x}}$ and $\vb{\hat{p}}$, i.e., $\hat{x}_i = \sum_{j \in S_i} x_j$ and $\hat{p}_i = \sum_{j \in S_i} p_j$ where $j \in S_i$ whenever bidder $i$ is associated with $\hat{b}_j$.
Also, we require the seller's mechanism $\Pi_S$ to allocate the item for some $\hat{b}_1,\ldots,\hat{b}_n$ at least (i.e., the seller should not always keep the item), and not allocate the item and charge no one if all bids are zero.

The intermediary, on the other hand, chooses a mechanism $\Pi_I=(\vb{\hat{b}},\vb{q})$ that maps bid-message pairs\footnote{By using a message interface, we essentially let the bidders, instead of the intermediary, choose what they want to forward. It is natural since bidders in our model always know the exact structure of both the intermediary's and the seller's mechanisms. Also, without this, in both intermediary-first model and simultaneous model, it would be unclear for the intermediary what to forward to the seller without seeing the seller's mechanism, introducing unnecessary complications.} $(b_1,m_1),\ldots,(b_n,m_n)$ from $n$ bidders to a list of bids $\vb{\hat{b}}((b_1,m_1),\ldots,(b_n,m_n)) \in \mathbb{R}_{\ge 0}^n$ forwarded to the seller and payment $\vb{q}((b_1,m_1),\ldots,(b_n,m_n)) \in \mathbb{R}_{\ge 0}^n$, and charge each bidder $i \in [n]$ a payment $q_i$.

The seller's mechanism $\Pi_S=(\vb{x},\vb{p})$ and the intermediary's mechanism $\Pi_I=(\vb{\hat{b}},\vb{q})$ together induces a mechanism $\Pi^*=(\vb{x^*},\vb{p^*})$ that maps bidders' bid-message pairs $((b_1,m_1),\ldots,(b_n,m_n))$ to the allocation 
$$\vb{x^*}((b_1,m_1),\ldots,(b_n,m_n)) = \vb{\hat{x}}(\vb{\hat{b}}((b_1,m_1),\ldots,(b_n,m_n)))$$
with the corresponding payment 
$$\vb{p^*}((b_1,m_1),\ldots,(b_n,m_n)) = \vb{\hat{p}}(\vb{\hat{b}}((b_1,m_1),\ldots,(b_n,m_n))) + \vb{q}((b_1,m_1),\ldots,(b_n,m_n)).$$

For bidder $i \in [n]$ with value $v_i$, her utility will be $x^*_i v_i - p^*_i.$
Bidders can decide their actions after knowing the exact structure of $\Pi^*$ (they have complete knowledge of $\Pi_S$ and $\Pi_I$), and they will always participate in $\Pi^*$ strategically.
For any $\Pi^*$ that has a Bayesian Nash equilibrium $(s_1,\ldots,s_n)$ where $s_i$ maps every $v_i$ to a bid-message pair $(b_i,m_i)$, we can define the seller's and the intermediary's revenue defined as:
\begin{align*}
    \SRev_{\Pi_S,\Pi_I}^s &= \E_{v_i \sim F_i}\left[\sum_{i \in [n]} p_i(\vb{\hat{b}}(s_1(v_1),s_2(v_2),\ldots,s_n(v_n)))\right], \\
    \IRev_{\Pi_S,\Pi_I}^s &= \E_{v_i \sim F_i}\left[\sum_{i \in [n]} q_i(s_1(v_1),s_2(v_2),\ldots,s_n(v_n))\right].
\end{align*}

Throughout our study of different variants of the game, it is reasonable to require $\Pi^*$ to satisfy \emph{ex-post individual rationality} for some strategy profile $(s_1,\ldots,s_n)$, i.e., every bidder $i$ following $s_i$ should never end up with a negative utility for any specific valuation $v_i$ and regardless of the strategies adopted by others.
Consequently, the intermediary can charge at most one bidder, as there is only a single indivisible item to allocate after all.
In other words, the intermediary's revenue always comes from a single bidder.
For this reason, we make the following assumption about the intermediary's mechanism. We refer the reader to Appendix~\ref{appendix:assumption} for additional discussion on this assumption.

\begin{assumption}\label{assumption:exclusive-right-dictate}
    $\Pi_I$ is an auction that sells the exclusive right to dictate the bids forwarded to $\Pi_S$. The allocation and pricing rule of $\Pi_I$ can only depend on bidder's bids $(b_1,\dots, b_n)$, and not their messages $(m_1,\dots, m_n)$.
\end{assumption}
 
Formally, for any $\Pi_I=(\vb{\hat{b}},\vb{q})$, we assume there exists an allocation indicator $\vb{y}(b_1,b_2,\ldots,b_n) \in \{0,1\}^n$ such that there exists at most a single bidder $i^*\in[n]$ with $y_{i^*}=1$, and payment $q_j = 0$ for any $j \ne i^*$. Whenever such bidder $i^*$ exists, the forwarded bids $\vb{\hat{b}}$ are dictated by bidder $i^*$, i.e., $\vb{\hat{b}}=m_{i^*}$.
Therefore, as long as the seller's mechanism $\Pi_S$ allocates to some bid in $\vb{\hat{b}}$, it will be the bidder $i^*$ who actually get the item.

Meanwhile, since we have assumed bidders know the exact structure of $\Pi^*$, they should always be able to use the list of bids that leads to a minimum payment whenever profitable.\footnote{Assuming seller's mechanism $\Pi_S$ is deterministic (we will relax this assumption in Section~\ref{subsec:randomized-mechanisms}).}
Formally, we can define the minimum achievable for $\Pi_S$ as follows:

\begin{definition}
    For any deterministic seller's mechanism $\Pi_S=(\vb{x},\vb{p})$, let $\omega(\Pi_S) \in \mathbb{R}_{\ge 0}$ be the \emph{minimum achievable payment} for $\Pi_S$ for a bidder who dictates all $n$ bids $\hat{b}_1,\hat{b}_2,\ldots,\hat{b}_n$.
    Formally, 
    $$\omega(\Pi_S)=\min_{\vb{\hat{b}} \in \mathbb{R}_{\ge 0}^n} \{\sum_{i \in [n]} p_i(\vb{\hat{b}}) \mid \sum_{i \in [n]} x_i(\vb{\hat{b}}) = 1\}.$$ 
    Let $\vb{\hat{b^*}}(\Pi_S)$ be the corresponding minimizer.
    Note that the minimum always exists since we require $\Pi_S$ sells the item for some bids.
\end{definition}

To summarize, the seller-intermediary game involves a seller, who chooses a mechanism $\Pi_S=(\vb{x},\vb{p})$, and an intermediary, who chooses another mechanism $\Pi_I=(\vb{\hat{b}},\vb{q})$. For a deterministc seller's mechanism $\Pi_S$, $\omega(\Pi_S)$ is the minimum achievable payment to get the item. The intermediary sells the exclusive right to dictate bids in seller's auction, and it is a dominant strategy for the winner in the intermediary's auction to dictate bids in the seller's auction in such a way that obtains the item with the lowest amount of payment possible.
Suppose $\vb{s}$ is bidders' strategy in a Bayesian Nash Equilibrium, we can then rewrite the revenue for the seller and the intermediary as,
\begin{align*}
    &\SRev_{\Pi_S,\Pi_I}^{\vb{s}}= \E_{v_i \sim F_i}\left[\omega(\Pi_S) \cdot \sum_{i\in [n]} x_i(\vb{\hat{b}}(\vb{s}(\vb{v})))\right], \\
    &\IRev_{\Pi_S,\Pi_I}^{\vb{s}} = \E_{v_i \sim F_i}\left[\sum_{i \in [n]} q_i(\vb{s}(\vb{v}))\right].
\end{align*}
where $\vb{s}(\vb{v}) = (s_1(v_1),s_2(v_2),\ldots,s_n(v_n))$.

\begin{theorem}
\label{thm:det-seller-post-price}
    Under Assumption~\ref{assumption:exclusive-right-dictate}, any deterministic mechanism $\Pi_S$ that the seller uses is equivalent to a posted-price mechanism with price $p=\omega(\Pi_S)$, denoted as $\Pi_S^{\AP(p)}$. For any Bayes-Nash equilibrium $\vb{s}$ of the game $(\Pi_S, \Pi_I)$, there exists a corresponding Bayes-Nash equilibrium $\vb{s'}$ of the game $(\Pi_S^{\AP(p)}, \Pi_I)$, such that $\SRev^{\vb{s'}}_{\Pi_S^{\AP(p)},\Pi_I}=\SRev^{\vb{s}}_{\Pi_S,\Pi_I}$ and $\IRev^{\vb{s'}}_{\Pi_S^{\AP(p)},\Pi_I}=\IRev^{\vb{s}}_{\Pi_S,\Pi_I}$, regardless of the intermediary's mechanism $\Pi_I$.
\end{theorem}
\begin{proof}
    For any mechanism $\Pi_S$, consider the posted-price mechanism $\Pi_S^{\AP(p)}$, where $p=\omega(\Pi_S)$. Consider an arbitrary mechanism chosen by the intermediary $\Pi_I = (\vb{\hat{b}},\vb{q})$. 
    Let $\vb{v}=(v_1,\dots, v_n)$ be bidders' values for the item. Then bidders' values for getting the exclusive right to dictate bids in seller's auction are $(\max\{v_1-p,0\},\dots, \max\{v_n-p\})$ under both $\Pi_S$ and $\Pi_S^{\AP(p)}$. For any Bayes-Nash equilibrium $\vb{s}$ of the game $(\Pi_S, \Pi_I)$, there exists a corresponding Bayes-Nash equilibrium $\vb{s'}$ of the game $(\Pi_S^{\AP(p)}, \Pi_I)$, where any bidder whose strategy includes a message that would obtain the item from the seller in $\vb{s}$ with minimum payment (such as $\vb{\hat{b^*}}(\Pi_S)$) replaces the message by $\vb{\hat{b^*}}(\Pi_S^{\AP(p)})=(p,0,\dots, 0)\in \mathbb{R}_{\ge 0}^n$ in $\vb{s'}$,  and any bidder whose strategy includes a message that does not obtain the item from the seller replaces the message by $(0,0,\dots, 0)\in \mathbb{R}_{\ge 0}^n$. We keep the bidders' bids in the intermediary's auction the same. Since the bidder's bids in the intermediary's auction are the same, the intermediary's revenue is the same under $\vb{s}$ and $\vb{s'}$ in the two games, and the allocation probability is also the same. This implies the seller's revenue is also the same under $\vb{s}$ and $\vb{s'}$ in the two games.
\end{proof}

\subsubsection{Variants of the Game}

In the following sections, we study the following variants of the seller-intermediary game:

\begin{enumerate}
    \item \textbf{Seller-first model} (\Cref{sec:seller-first}): The seller leads a Stackelberg game, i.e., the seller first commits to a mechanism $\Pi_S$, and the intermediary chooses a mechanism $\Pi_I$ after seeing $\Pi_S$.
    \item \textbf{Intermediary-first model} (\Cref{sec:bp-first}): The intermediary leads a Stackelberg game, i.e., the intermediary first commits to a mechanism $\Pi_I$, and the seller chooses a mechanism $\Pi_S$ after seeing $\Pi_I$.
    \item \textbf{Simultaneous model} (\Cref{sec:simultaneous}): The seller and the intermediary moves simultaneously.
\end{enumerate}

%% file: seller_first.tex
\section{Seller-First Model}

\label{sec:seller-first}

In this section, we study the seller-first model of the seller-intermediary game, where the seller first commits to a mechanism $\Pi_S$, and the intermediary is able to best respond with a mechanism $\Pi_I$.
We first characterize intermediary's best responses to any seller's mechanism in \Cref{subsec:seller-first-best-response}.
Then, we give bounds on the revenue guarantees for both the seller and the intermediary under different distributional assumptions in \Cref{subsec:seller-first-revenue}.
Furthermore, we also discuss and relax the assumption of deterministic mechanisms in \Cref{subsec:randomized-mechanisms}, showing that randomization does not help the seller achieve a higher revenue in this setting.
Proofs in this section are deferred
to \Cref{appendix:missing-proof-seller-first}.

\subsection{Best Response for the Intermediary}
\label{subsec:seller-first-best-response}

We first characterize the intermediary's best response to any seller's mechanism $\Pi_S$ assuming $F_1,F_2,\ldots,F_n$ are regular.
We show the best response $\Pi_I^{\BR}(\Pi_S)$ for the intermediary remains Myerson's optimal auction, but with a different reserve that depends on the minimum achievable payment $p_0=\omega(\Pi_S)$ of the seller's mechanism $\Pi_S$.

Note that in the seller-first model, once the seller commits to a mechanism \(\Pi_S\),
the intermediary's problem becomes a standard single-parameter auction design
problem. Indeed, by Theorem~\ref{thm:det-seller-post-price}, the seller's
mechanism is equivalent to a posted price \(p_0=\omega(\Pi_S)\). Hence bidder
\(i\)'s value for winning the intermediary's auction is $\max\{v_i-p_0, 0\}$. Therefore, by the revelation principle, we may restrict
attention without loss of generality to direct truthful mechanisms for the
intermediary's auction: for any Bayes-Nash equilibrium of any indirect
intermediary mechanism, there is a truthful direct mechanism that induces the
same allocation and payments, and hence the same expected seller and intermediary revenue.

\begin{theorem}
\label{thm:seller-first-best-response}
Assume $F_1,F_2,\ldots,F_n$ are regular distributions.
For any mechanism $\Pi_S$ committed by the seller with a minimum payment $p_0=\omega(\Pi_S)$, the intermediary's best response is the mechanism $\Pi_I^{\BR}(\Pi_S)=(\vb{b}^{\BR},\vb{q}^{\BR})$ described as follows:
\begin{itemize}
    \item The intermediary runs the Myerson's optimal auction with a personalized reserve of $\phi_i^{-1}(p_0)-p_0$ for bidder $i \in [n]$ with truthful bid $b_i=\max(v_i-p_0,0)$, and let $\vb{q}^{\BR}$ be the corresponding payment.
    In other words, the bidder $i^*$ with highest virtual value $\phi_i(v_{i^*}) \ge p_0$ wins with a payment of $q_{i^*}^\BR=\phi_i^{-1}(\max\{p_0,\max_{j \ne i} \phi_j(v_j)\})-p_0$.
    
    \item If there exists a bidder $i^*\in [n]$ who wins, let forwarded bids $\vb{b}^{\BR}=\vb{\hat{b}^*}(\Pi_S)$ (i.e., the list of bids that achieves a minimum payment $p_0$ in $\Pi_S$) be the bids associated with bidder $i^*$. 
    The bidder $i^*$ will get the item from the seller at a price of $p_0$.

    \item Otherwise, let forwarded bids $\vb{b}^{\BR}=(0,0,\ldots,0)$ and no sale occurs.
\end{itemize}
\end{theorem}

\subsection{Revenue Guarantees}
\label{subsec:seller-first-revenue}

\input{seller_first_approx_ratios}

\subsection{Randomized Mechanisms}
\label{subsec:randomized-mechanisms}

We have so far assumed that the seller's mechanism $\Pi_S$ is deterministic.
We next extend our result to sellers that commit to randomized mechanisms, showing that randomization does not help the seller.
Specifically, we consider a randomized mechanism $\widetilde{\Pi}_S$ that induces a menu of $m$ options $\{(\alpha_j,\beta_j)\}_{j \in [m]}$.
Each option $(\alpha_j,\beta_j)$ indicates the minimum price $\beta_j \ge 0$ to obtain the item with probability $\alpha_j \in [0,1]$.

\begin{theorem}\label{thm:randomization-does-not-help-seller}
    Assume $F_1,F_2,\ldots,F_n$ are regular distributions.
    Given that the intermediary always best responds, any randomized mechanism $\widetilde{\Pi}_S$ yields an expected revenue not larger than the optimal seller's expected revenue when using deterministic mechanisms.
\end{theorem}

We first characterize the revenue of a randomized seller's mechanism in \Cref{lem:randomized-seller-revenue}. We defer its proof to \Cref{appendix:missing-proof-seller-first}.

\begin{lemma}\label{lem:randomized-seller-revenue}
    Assume $F_1,F_2,\ldots,F_n$ are regular distributions.
    Given that the intermediary always best responds, for any seller's randomized mechanism $\widetilde{\Pi}_S$ with a menu $\{(\alpha_j,\beta_j)\}_{j \in [m]}$, the seller's expected revenue is
    $$\E[\beta_{j^*} \I[\alpha_{j^*}\overline{v}-\beta_{j^*} \ge 0]]$$
    where $\overline{v}=\max_{i \in [n]} \phi_i(v_i)$ and $j^*= \argmax_{j \in [m]} (\alpha_j \overline{v} - \beta_j)$.
\end{lemma}

Having obtained such a characterization, \Cref{thm:randomization-does-not-help-seller} can be proven by observing that the seller's revenue exactly corresponds to the revenue in another auction design setting, where the auctioneer is trying to sell an item to a single bidder using the same menu.
Since a deterministic posted-price is optimal for selling an item to a single bidder, we can conclude that randomization also does not help in the original setting.

%% file: seller_first_approx_ratios.tex
Assuming the intermediary always chooses the best response $\Pi_I^\BR(\Pi_S)$ to the seller's mechanism $\Pi_S$, we study how well the seller's revenue $\SRev_{\Pi_S,\Pi_I^\BR(\Pi_S)}$ can approximate the optimal revenue $\OptRev$ when there is no intermediary.
For general regular distributions, we give an explicit example where $\SRev_{\Pi_S,\Pi_I^\BR(\Pi_S)}$ can be arbitrarily smaller than $\OptRev$ for any seller's mechanism $\Pi_S$ (\Cref{thm:max-ssm-opt-seller-leader}).
On the other hand, for $\alpha$-strongly regular distributions with $0< \alpha \leq 1$, we prove the seller's revenue $\SRev_{\Pi_S^{\AP^*},\Pi_I^\BR(\Pi_S^{\AP^*})}$ can approximate $\OptRev$ to a constant factor when the seller uses some anonymous posted-price mechanism $\Pi_S^{\AP^*}$ (\Cref{thm:seller-first-alpha-strongly-regular-revenue}).

\begin{remark}[Intermediary's revenue in the seller-first model] \label{remark:intermediary-revenue}
    When there is a single bidder and its value distribution simply takes a known value $v$ (which is a MHR distribution), the seller run an anonymous posted-price mechanism with a price of $v-\varepsilon$ for any small $\varepsilon > 0$.
    In this case, the intermediary's best response is to allocate a bidder with a posted price $\varepsilon$.
    This implies that the intermediary's revenue $\IRev$ cannot approximate $\OptRev$ by any constant $\varepsilon>0$ in the worst case for MHR distributions in general.

    On the other hand, note that for some regular distributions, the intermediary can capture nearly all of the revenue (e.g., see \Cref{cor:example-move-second-rev-opt}).
\end{remark}

\subsubsection{Regular Distributions}

We first give an explicit example of regular distributions $F_{\varepsilon,H}$ that make the seller's revenue arbitrarily small compared to $\OptRev$. 
The regular distribution $F_{\varepsilon,H}$, parameterized by  $\varepsilon,H>0$, is supported on $[1,H]$ with CDF 
$$F_{\varepsilon,H}(x)=\begin{cases}1-\frac{1}{x^{1+\varepsilon}} & 1\le x < H, \\ 1 & x=H. \end{cases}$$
Note that although $F_{\varepsilon,H}(x)$ is not continuous at $x=H$, its virtual value function $\phi_{\varepsilon,H}$ is still well-defined as $\phi_{\epsilon,H}(x)=\frac{\varepsilon}{1+\varepsilon}x$ for all $1 \le x < H$ and $\phi_{\varepsilon,H}(H)=H$, and thus $F_{\varepsilon,H}$ is $\frac{\varepsilon}{1+\varepsilon}$-strongly regular. Additionally, for tie-breaking, we define the probability of sale when price is posted as $H$ to be $\Pr[v\geq H]=\frac{1}{H^{1+\epsilon}}$.

When there is a single bidder with valuation drawn from $F_{\varepsilon,H}$, we can characterize intermediary's best response as follows.

\begin{lemma}
\label{lem:example-BR}
    Suppose there is one bidder with valuation drawn according to $F_{\varepsilon,H}$ for some $\varepsilon,H>0$, and the seller uses a mechanism $\Pi_S$ with a minimum price $p=\omega(\Pi_S)$. The intermediary's best response $\Pi_I^\BR(\Pi_S)$ is a posted-price mechanism with price $r^*(p)=\min\{p/\varepsilon, H-p\}$.
\end{lemma}

Based on \Cref{lem:example-BR}, we show the seller's revenue in this instance will be much worse than $\OptRev$ for small $\varepsilon$.

\begin{theorem}
\label{thm:max-ssm-opt-seller-leader}
    Suppose there is one bidder with valuation drawn according to $F_{\varepsilon,H}$ for some $\varepsilon>0$ and $H>(1+1/\varepsilon)^{1+1/\varepsilon}$.
    Given that the intermediary always best responds, the seller's revenue $R_S=\SRev_{\Pi_S,\Pi_I^\BR(\Pi_S)}$ using any mechanism $\Pi_S$ is at most $(1+1/\varepsilon)^{-(1+\varepsilon)}$ (which is $O(\varepsilon)$ when $\varepsilon \to 0$) while $\OptRev=1$.
\end{theorem}

In contrast, the intermediary in this instance can get a nearly-optimal revenue for small $\varepsilon$.

\begin{corollary}
\label{cor:example-move-second-rev-opt}
    Suppose there is one bidder with valuation drawn according to $F_{\varepsilon,H}$ for some $\varepsilon>0$ and $H>(1+1/\varepsilon)^{1+1/\varepsilon}$.
    Given that the seller always choose $\Pi_S$ optimally, the best-responding intermediary's revenue $\IRev_{\Pi_S,\Pi_I^\BR(\Pi_S)}=\frac{1/\epsilon}{(1+1/\epsilon)^{1+\epsilon}}$ (which goes to $1$ when $\varepsilon\to 0$).
\end{corollary}

\begin{corollary}
\label{cor:lower-bound-strong-regularity}
    For any $0<\alpha<1$, there exists a bidder with valuation drawn from a $\alpha$-strongly regular distribution $F_{\alpha}$ such that no seller's mechanism can guarantee more than $c(\alpha)*\OptRev$ revenue for the seller where $c(\alpha)=\alpha^{{1}/(1-\alpha)}$. 
\end{corollary}

\subsubsection{$\alpha$-Strongly Regular Distributions}

We show that when the bidders' value distributions $F_1,F_2,\ldots,F_n$ are $\alpha$-strongly regular distributions for some $0<\alpha\le 1$ (MHR when $\alpha=1$), there exists an anonymous posted-price mechanism that guarantees the seller at least a constant fraction of the optimal revenue, whose dependence on $\alpha$ is almost tight. 

\begin{theorem}\label{thm:seller-first-alpha-strongly-regular-revenue}
    Assume $F_1,F_2,\ldots,F_n$ are $\alpha$-strongly regular distributions with $0< \alpha \le 1$.
    Given that the intermediary always best responds, there exists an anonymous posted-price mechanism $\Pi_S^{\AP^*}$ with an expected revenue $R_S=\SRev_{\Pi_S^{\AP^*},\Pi_I^\BR(\Pi_S^{\AP^*})}$ such that
    \begin{itemize}
        \item $R_S \ge (1-\frac{1}{e}) c(\alpha) \cdot \OptRev$ when $F_1,F_2,\ldots,F_n$ are identical, and
        \item $R_S \ge \frac{1}{2.62} c(\alpha) \cdot \OptRev$ when $F_1,F_2,\ldots,F_n$ are non-identical,
    \end{itemize}
    where $c(\alpha)=\alpha^{1/(1-\alpha)}$ for $0<\alpha<1$ and $c(1)=\frac{1}{e}$.
\end{theorem}
Together with \Cref{cor:lower-bound-strong-regularity}, this shows that the dependence on $c(\alpha)$ is tight up to a constant factor, for every $0<\alpha<1$.

To prove \Cref{thm:seller-first-alpha-strongly-regular-revenue}, we first compare the revenue of $\Pi_S^{\AP}$ with the revenue of the optimal anonymous posted-price mechanism, $\APRev$, as stated in \Cref{lem:alpha-strongly-regular-seller-revenue-to-ap}, then combining it with previous works that relate $\APRev$ to $\OptRev$ (\Cref{thm:ratio-ap-to-opt-iid}, \Cref{thm:ratio-ap-to-opt-non-iid}).

\begin{lemma}
\label{lem:alpha-strongly-regular-seller-revenue-to-ap}
Assume $F_1,F_2,\ldots,F_n$ are $\alpha$-strongly regular distributions with $0< \alpha \le 1$.
Given that the intermediary always best responds, there exists an anonymous posted-price mechanism $\Pi_S^{\AP^*}$ with an expected revenue $R_S=\SRev_{\Pi_S^{\AP^*},\Pi_I^\BR(\Pi_S^{\AP^*})}$ such that $R_S \ge c(\alpha) \cdot \APRev$ where $c(\alpha)=\alpha^{1/(1-\alpha)}$ for $0<\alpha<1$ and $c(1)=\frac{1}{e}$.
\end{lemma}

%% file: intermediary_first.tex
\section{Intermediary-First Model}
\label{sec:bp-first}

We next focus on the intermediary-first model, where the intermediary first commits to its mechanism $\Pi_I$, and the seller is able to best-respond with a mechanism $\Pi_S$. We make the assumption that the seller is using a deterministic mechanism $\Pi_S$, so by Theorem~\ref{thm:det-seller-post-price}, the seller uses a posted price mechanism $\Pi_S=\Pi_{S}^{\AP(p)}$ for some price $p$.
Some proofs in the section are deferred to \Cref{appendix:missing-proof-bp-first}.

\subsection{Best Response for the Seller}
\label{sec:intermediary-first-best-response}

Suppose the intermediary commits to some mechanism $\Pi$, and the seller uses an anonymous posted-price mechanism with price $p$. Let $B_{\Pi,p}$ denote the probability that the right to dictate the bids is sold in intermediary's auction $\Pi$, when the seller posts a price $p$. Then we state the best response of the seller to the intermediary's auction in the following theorem.

\begin{theorem}
\label{thm:active-bp-best-response-new-model}
When the intermediary commits to a mechanism $\Pi$, the posted-price that maximizes seller's revenue is given by
\[p^* = \argmax_{p \geq 0}\{p\cdot B_{\Pi,p}\}.\] Let $\Pi_{S}^{\BR}(\Pi_I)$ denote the seller's best response to the intermediary's auction $\Pi_I$. Then $\Pi_{S}^{\BR}(\Pi_I) = \Pi_{S}^{\AP(p^*)}$.
\end{theorem}
\begin{proof}
The seller gets a revenue of $p$ whenever the intermediary sells the right to dictate bids to a bidder, which happens with probability $B_{\Pi,p}$.
\end{proof}

\subsection{Revenue Guarantees}
\label{sec:intermediary-first-approximation_ratios}

In this section, we present lower and upper bounds for the ratios $\SRev/\OPT$ and $\IRev/\OPT$.

\subsubsection{Regular Distributions}

For regular distributions, $\IRev/\OPT$ can be arbitrarily small even for a single bidder.

\begin{theorem}
\label{thm:intermediary-revenue-upper-bound}
There exists a regular distribution such that for any $\Pi_I$, $\IRev_{\Pi_{S}^{\BR}(\Pi_I),\Pi_I}/\OPT$ is arbitrarily small even for a single bidder.
\end{theorem}

\begin{proof}
When there is a single bidder and the intermediary's auction is deterministic, any intermediary's auction is equivalent to a posted-price mechanism (for the same reason as the proof of \Cref{thm:det-seller-post-price}). By \Cref{thm:det-seller-post-price}, we can also assume the seller is running a posted price mechanism without loss of generality. Thus in this case, the intermediary' role in the intermediary-first model is exactly the seller's role in the seller-first model and the two models are symmetric. Thus our result for the seller in Theorem~\ref{thm:max-ssm-opt-seller-leader} applies to the intermediary in this model.
\end{proof}

\subsubsection{MHR Distributions}

For i.i.d.\ MHR distributions, on the other hand, we show that running a simple posted price mechanism can guarantee the intermediary a constant fraction of the optimal revenue.

\begin{theorem}
\label{thm:bp-first-mhr}
    Assume $F_1,F_2,\ldots,F_n$ are identical MHR distributions.
    Given that the seller always best responds, there exists an intermediary's mechanism $\Pi_I$ that guarantees an expected revenue of $\IRev_{\Pi_{S}^{\BR}(\Pi_I),\Pi_I} \geq \frac{1}{e}(1-\frac{1}{e}) \OPT$. 
\end{theorem}

Our proof relies on the following well-known property of maximum over MHR distributions.

\begin{lemma}\label{lem:maximum-mhr-remains-mhr}
    For any MHR distribution with CDF $F$ and $n$ i.i.d.\ random variables $X_1,X_2,\ldots,X_n \sim F$. Let $F_{\max}$ denote the distribution of their maximum $Y=\max(X_1,X_2,\ldots,X_n)$. Then $F_{\max}$ is also a MHR distribution and $F_{\max}=F^n$.
\end{lemma}

\begin{remark}[Seller's revenue in the intermediary-first model]
\label{remark:seller-revenue}
When there is a single bidder and its value distribution simply takes a known value $v$ (which is a MHR distribution), the intermediary can set its posted price to be $v-\epsilon$ for any small $\epsilon > 0$, and the seller's best response is to set a posted price $\epsilon$. Thus in this case $\SRev/\OPT$ is arbitrarily small. On the other hand, for some regular distributions, the seller can capture all of the revenue, \ie, $\SRev/\OPT$ can be arbitrarily close to $1$ (see Corollary~\ref{cor:example-move-second-rev-opt} and swapping intermediary and seller's roles).
\end{remark}

%% file: simultaneous.tex
\section{Simultaneous model}
\label{sec:simultaneous}

In this section, we consider the third model of the seller-intermediary game, where the seller and the intermediary move simultaneously, that is, they reveal their mechanisms $\Pi_S$ and $\Pi_I$ simultaneously, and $(\Pi_S,\Pi_I)$ form a Nash equilibrium. Some proofs in this section are deferred to \Cref{appendix:missing-proof-simultaneous}.

\begin{theorem}
    Consider the seller-intermediary game under the simultaneous model with a single bidder. Then there exist bidder valuations where there are infinite number of pure Nash equilibrium for the game.
\end{theorem}

\begin{proof}
   Let the single bidder has a constant value $v>0$ for the item. Since we are in the single bidder case, both the seller and the intermediary's best response mechanism is a posted-price mechanism, by \Cref{thm:det-seller-post-price} and \Cref{thm:seller-first-best-response}. For simplicity, we let a pair $(p,r)$ denote the posted-price mechanisms of the seller and the intermediary, respectively. Then it is easy to see that any pair of strategies $(p, r=v-p)$ for $p\in [0,v]$ is a Nash equilibrium, since the maximum combined price to set is $v$ and setting anything lower than $v$ means they are not best responding to each other.
\end{proof}
The above examples shows that using Nash as an equilibrium solution concept for the seller-intermediary game is somehow undesirable. Additionally, our main result of the section show that there exists $\alpha$-strongly regular distributions for some $0<\alpha<1$ where every Nash equilibrium yields very bad revenue for both the seller and the intermediary compared to the sell-first model or the intermediary-first model.

\begin{theorem}
\label{thm:sim-unbdd}
Consider the seller-intermediary game under the simultaneous model with a single bidder. There exists $\alpha$-strongly regular value distributions for some $0<\alpha<1$ for the bidder such that for any Nash equilibrium strategies of the seller and the intermediary, the seller and the intermediary's revenue are unboundedly worse than their revenue in any equilibrium where the seller moves first or the intermediary moves first.

More specifically, suppose the single bidder has value distribution $F_{\varepsilon,H}$ for some $\varepsilon>0$ and $H>(1+1/\varepsilon)^{1+1/\varepsilon}$. Then for any Nash equilibrium mechanisms $(\Pi_S, \Pi_I)$ of the seller-intermediary game, $\max \{\SRev_{\Pi_S, \Pi_I}, \IRev_{\Pi_S, \Pi_I}\}\leq \frac{1}{H^{\epsilon}}$. However, in the seller-first model or the intermediary-first model,
the first mover gets a revenue of $\frac{1}{(1+1/\epsilon)^{1+\epsilon}}$, and the second mover gets a revenue of $\frac{1/\epsilon}{(1+1/\epsilon)^{1+\epsilon}}$. This implies a minimum revenue of $\min \{\SRev,\IRev\}=\frac{1}{(1+1/\epsilon)^{1+\epsilon}}$ for the optimal mechanisms used in both models. Thus $\frac{\max \{\SRev_{\Pi_S, \Pi_I}, \IRev_{\Pi_S, \Pi_I}\}}{\min \{\SRev,\IRev\}}\leq\frac{(1+1/\epsilon)^{1+\epsilon}}{H^\epsilon}$. The ratio goes to 0 when $H\to \infty$.
\end{theorem}

%% file: appendix_related_work.tex
\section{Extended Related Work}
\label{sec:appendix_related-work}

\subsection{Auctions with Intermediaries}
\label{sec:related-work-awi}

Feldman et al. consider intermediaries in the context of online advertising auctions, where the advertisers (bidders) rely on an ad exchange (intermediary) with access to a variety of publishers (sellers) to bid on their behalf~\cite{FMM+10}.
There is a single seller and multiple intermediaries, each with an equal-sized set of captive bidders that have i.i.d. MHR values.
Similar to our seller-first model, the seller first announces a second-price auction with reserve.
Afterwards, the intermediaries run a Vickrey auction with reserve to sell a contingent good that will be acquired from the seller.
When the auction is restricted to be DSIC, they adopt a posted-price mechanism.
Feldman et al. analyze the selection of intermediary's reserve prices and show that the reserve prices are randomized over an interval $[\mathsf{LB}, \mathsf{UB}]$.
When there is a single buyer per intermediary, as the number of intermediaries grows, $\mathsf{UB}$ increases, and the mass distribution moves towards $\mathsf{UB}$, while the seller's reserve price decreases (but remains strictly positive).
In contrast, our seller-first model focuses on a single intermediary, but with multiple captive bidders.

Stavrogiannis et al. extend the model in~\cite{FMM+10} by considering two \emph{competing} intermediaries and non-captive bidders~\cite{SGP12,SGP13}.
They show the existence of infinitely many Bayesian Nash equilibria and a symmetric sub-game perfect Nash equilibrium, where two competing intermediaries set reserve prices equal to the expected value of the second highest bid.
When these reserve prices are sufficiently different, the unique pure NE for all bidders is to select the low-reserve intermediary.
Otherwise, the bidders in the low-value interval choose the low-reserve intermediary, those in the middle follow a strictly mixed strategy, and those in the high-value interval all go after either the high-reserve or the low-reserve intermediary.
Unlike these works, our work considers a single intermediary with captive bidders.

Stavrogiannis et al. later expand on their results in~\cite{SGP12,SGP13} by also analyzing a single intermediary with captive bidders~\cite{SGP14}.
Here, the intermediaries can run three mechanisms: first-price sealed-bid auction, and two variations of the Vickrey auction termed pre- and post-award Vickrey auctions, depending on the timing of the payments to the intermediary.
In this setting, the optimal reserve price for the seller is shown to increase with the number of buyers while the welfare decreases.
Although the setting of a single intermediary, post-award Vickrey auction is similar to our seller-first model, besides investigating other models, our work provides explicit bounds on welfare and the revenues of the involved parties given a wider class of distributions.

Balseiro et al. consider intermediaries that offer a menu of contracts to the bidders~\cite{BC17}.
Here, the bidders also have an outside option with a worse surplus.
Each contract entails a particular bidding policy on behalf of the bidders for an up-front payment.
Balseiro et al. identify the optimal contract, which is shown to improve market efficiency compared to the outside option, despite the intermediary's extraction of profit.
In a similar setting, Balseiro et al. analyze fixed-percentage ($\alpha$) revenue sharing schemes between the seller and intermediaries, where the intermediaries pay the seller at least a seller-declared opportunity cost for each item sold~\cite{BLM17}.
In our work, we do not consider any outside option for the exchange between the seller and bidders.

\subsection{Multi-level Mechanisms}

Babaioff et al. analyze the effect of mediators that control the flow of information from strategic agents associated with them to a mechanism that aims to maximize the total utility of all agents~\cite{BFT16}.
Each mediator in turn aims to optimize the same utility, but only for its own agents.
They focus on the goal of selecting a vertex on a publicly known tree, where the social cost is the total distance of the parties to the vertex.
Unlike our work, each intermediary aims to optimize the utility of its agents and is `completely benevolent'.

Bahrani et al. analyze two-level DSIC auctions, where the participants of the upper-level mechanism are groups of bidders~\cite{BGR23}.
Each group in turn distributes the access to the acquired item and the aggregate payment demanded from the group through a lower-level mechanism.
In contrast, our work considers individual bidders whose interaction with the seller is mediated by a single strategic party that aims to extract revenue from the auction.

\subsection{Auctions in Blockchains}

Auctions have found widespread use on blockchains for items such as debts and non-fungeable tokens (NFTs)~\cite{maker_auctions,a16zauction}.
They are also used to allocate block space in return for transaction fees through transaction fee mechanisms (TFMs).
In this context, Huberman et al. analyzed Bitcoin's TFM as a payment system that avoids monopoly pricing due to the free entry of miners~\cite{HLM21}.
Lavi et al. proposed a monopolistic auction for Bitcoin to achieve strategyproofness as the number of bids increases (their main conjecture was proven by Chi-Chih-Yao)~\cite{LSZ22,yao20}.
Roughgarden introduced the core notions of incentive compatibility for TFMs and analyzed the mechanisms proposed for Ethereum in terms of these properties~\cite{roughgarden21}.
Chung and Shi later proposed a burning second-price auction that achieves these notions in a setting where malicious validators eventually suffer financial losses for strategic bidding~\cite{CS23}.

Blockchains have also been shown to support revenue maximizing, credible, and truthful auctions~\cite{AL18,FW20,EFW22,CFK23} for items besides block space.
However, these guarantees crucially require the blockchain to behave like a public, append-only, and censorship-resistant ledger.
In practice, blocks are created by \emph{block proposers}, which have monopoly power over their contents and can exclude, insert, and reorder transactions for financial gain.
Quantified under the term \emph{maximal extractable value} (MEV), the total amount of proposers' gain obtained this way exceeded $1.8$ billion USD between September 2020 and 2024 on Ethereum, the largest blockchain that supports on-chain auctions~\cite{mev_estimate}.

Besides e-commerce and agentic trade, our work is also applicable to blockchain auctions for items other than block space; since in the case of block space auctions, the seller is the block proposer itself, and there is no intermediary.

%% file: appendix.tex
\section{Summary of Results}
\label{appendix:summary-results}
We summarized our revenue guarantees and inapproximability results under different models and bidder distributions in Table~\ref{tab:results-stakel}. Here we use $R_S$ and $R_I$ as a shorthand notation for $\SRev$ and $\IRev$. We let $c(\alpha)=\alpha^{1/(1-\alpha)}$ for $0<\alpha<1$ and $c(1)=\frac{1}{e}$.
\begin{table}[h!]
  \fontsize{9pt}{9pt}\selectfont
  \centering
  \caption{Comparison of Revenue Approximations under Different Models and Distributions}
  \label{tab:results-stakel}
  \begin{tabularx}
  {\textwidth}{@{} l  l  l  @{}}
    \toprule
    {} 
      &
    \textbf{Seller-First} 
      & \textbf{Intermediary-First}  \\
    \midrule
    Regular
      & \shortstack[l]{$R_S/\OptRev$ can approach 0 
      (\Cref{thm:max-ssm-opt-seller-leader})\\
      $R_I/\OptRev$ can approach 0 (\Cref{remark:intermediary-revenue})}
      & \shortstack[l]{$R_S/\OptRev$ can approach 0
      (\Cref{remark:seller-revenue})\\
      $R_I/\OptRev$ can approach 0
     (\Cref{thm:intermediary-revenue-upper-bound})} 
      \\
    \midrule
   MHR
      & \shortstack[l]{$R_I/\OptRev$ can approach 0 (\Cref{remark:intermediary-revenue})}
      & \shortstack[l]{$R_S/\OptRev$ can approach 0 (\Cref{remark:seller-revenue})}
      \\
    \midrule
    \shortstack[l]{$\alpha$-strongly regular\\ for $0<\alpha\leq 1$}
    & \shortstack[l]{identical $F_i$: $R_S/\OptRev\geq (1-1/e)c(\alpha)$\\  non-identical $F_i$: $R_S/\OptRev\geq (1/2.62)c(\alpha)$\\
    (\Cref{thm:seller-first-alpha-strongly-regular-revenue})}
      & \shortstack[l]{identical $F_i$ and $\alpha=1$: \\$R_I/\OptRev\geq (1/e)(1-1/e)$
      \\
      (\Cref{thm:bp-first-mhr})}\\
    \bottomrule
  \end{tabularx}
\end{table}

\begin{table}[h!]
 \fontsize{9pt}{9pt}\selectfont
  \centering
  \label{tab:results-sim}
  \begin{tabularx}{\textwidth}{@{} l  c@{}}
    \toprule
    {} 
      & \textbf{Simultaneous} \\
    \midrule
    \shortstack[l]{$\alpha$-strongly regular\\ for $0<\alpha< 1$}
      & \shortstack[l]{$R_I, R_S$ in Nash Equilibrium can be unboundedly worse\\ than in Seller-first and Intermediary-first models (\Cref{thm:sim-unbdd}).} 
      \\
    \bottomrule
  \end{tabularx}
\end{table}

\section{Discussion on Assumption~\ref{assumption:exclusive-right-dictate}}
\label{appendix:assumption}
In this section, we offer some additional justifications for Assumption~\ref{assumption:exclusive-right-dictate}. Assumption~\ref{assumption:exclusive-right-dictate} has two components: (a) that $\Pi_I$ sells exclusive rights to dictate the messages forwarded to $\Pi_S$ and (b) that $\Pi_I$ depends only on bids and not on forwarded messages. 

Consider (a) and (b) first from the perspective of an Intermediary who moves First. When the Intermediary moves First, the Seller has not even specified an auction (or message space) yet. Therefore, it would be quite unnatural for the Intermediary to ascribe semantic meaning to any messages of the bidders, and therefore the intermediary would naturally satisfy both (a) and (b).\footnote{(b) is straight-forward. To conclude (a), observe that without semantic meaning on messages, the Intermediary has no idea how the Seller's auction would treat the concatenation of $m_1$ and $m_2$, and so it is unnatural to forward multiple messages. On the other hand, if each Bidder $i$ knows that their message $m_i$ will be \emph{exclusively} sent to the Seller's auction, then Bidder $i$ knows exactly their value for winning.} Therefore, a First-Moving Intermediary \emph{should} sell exclusive rights (as otherwise, they don't even know what they're selling, because they don't know the Seller's auction). That is, although the intermediaries we consider in this work (\eg, blockchain intermediaries) are not currently selling exclusive rights, we argue that any First-Moving Intermediaries \emph{should} (and therefore, we also expect First-Moving Intermediaries to eventually converge to doing so).

Finally, consider (a) and (b) from the perspective of an Intermediary that moves Second. There is ultimately a single item for sale, and therefore the revenue-maximizing auction for the Intermediary will always award the item to the same bidder. As soon as the Intermediary determines the winning bidder, the Intermediary and Bidder are aligned in wanting to choreograph the forwarded bids to win the item from Seller's auction at the lowest price possible (because the Bidder cares only about their total payment, the Intermediary can collect more revenue by forwarding less to Seller). Therefore, a revenue-optimal Intermediary would choose to satisfy (a) and (b) anyway.

\section{Preliminaries for Revenue Maximization}
\label{appendix:prelim-for-revenue-maximization}

\paragraph{Virtual Values.}
For a continuous distribution $F$ with PDF $f$, the virtual value of $v$ is defined as $\phi(v) = v - \frac{1-F(v)}{f(v)}$.\footnote{Note that for a distribution truncated down to some $H$ (which makes it discontinuous at $H$), its virtual value function $\phi$ at $H$ remains defined as $\phi(H)=H$.}
We will assume PDF exists for all bidders' value distributions $F_1,F_2,\ldots,F_n$, and we denote the virtual value functions of the bidders $i \in [n]$ by $\phi_i(\cdot)$.

\paragraph{$\alpha$-Strongly Regular Distributions.}

Unless noted otherwise, we will assume all distributions $F_1,F_2,\ldots,F_n$ are regular, and we will be working with $\alpha$-strongly regular distributions a lot.

\begin{definition}[$\alpha$-strongly regular distribution]
    Let $F$ be a distribution supported on $[a,b]$ where $0 \le a < \infty$ and $a \le b \le \infty$. Let $\phi(v)$ be the corresponding virtual value function. 
    For any $\alpha \ge 0$, $F$ is \emph{$\alpha$-strongly regular} if for all $a \le x \le y \le b$, 
    $$\phi(y)-\phi(x) \ge \alpha(y-x).$$
    A distribution is \emph{regular} if it is $0$-strongly regular, and \emph{monotone hazard rate (MHR)} if it is $1$-strongly regular.
\end{definition}

Here is a useful property of the virtual value function of a $\alpha$-strongly regular distribution.

\begin{lemma}[Lemma~4.1, Lemma~4.2 in~\cite{RT14}]
\label{lem:property-of-alpha-strongly-regular}
    Let $F$ be any $\alpha$-strongly regular distribution for $0 < \alpha \le 1$ with a virtual value function $\phi_F(\cdot)$.
    We have $\Pr_{v \sim F}[v \ge \phi_F^{-1}(0)] \ge c(\alpha)$ where $c(\alpha)=\alpha^{1/(1-\alpha)}$ for $0<\alpha<1$ and $c(1)=\frac{1}{e}$. In other words, $1-F(\phi_F^{-1}(0)) \ge c(\alpha)$.
\end{lemma}

\begin{corollary}\label{cor:property-of-alpha-strongly-regular}
    Let $F$ be any $\alpha$-strongly regular distribution for $0 < \alpha \le 1$ with a virtual value function $\phi_F(\cdot)$.
    For any $p$ in the support of $F$, 
    we have $\Pr_{v \sim F}[v \ge \phi_F^{-1}(p) \mid v \ge p] \ge c(\alpha)$ where $c(\alpha)=\alpha^{1/(1-\alpha)}$ for $0<\alpha<1$ and $c(1)=\frac{1}{e}$.
    In other words, $1-F(\phi_F^{-1}(p)) \ge c(\alpha) \cdot (1-F(p))$ 
\end{corollary}

\begin{proof}
    Suppose $F$ is $\alpha$-strongly regular with $0 < \alpha \le 1$. Consider the distribution $G=(F-p)_{\ge 0}$ with CDF defined as \[G(x)=\frac{1}{1-F(p)}(F(x+p)-F(p)),\] i.e. $G(x)=\Pr_{v\sim F}[v\leq x+p\mid v\geq p]$. Then $G$ is also an $\alpha$-strongly regular distribution since $\phi_G(v)=\phi_F(v+p)-p$ for all $v\ge 0$.
    Also note that $\phi_G^{-1}(0)=\phi_F^{-1}(p)-p$.
    Therefore, 
    \begin{align*}
    \Pr_{v \sim G}[v \ge \phi^{-1}_G(0)] 
    &= \Pr_{v \sim G}[v \ge \phi^{-1}_F(p) - p] \\
    &= \Pr_{v \sim F}[v-p \ge \phi^{-1}_F(p)-p \mid v \ge p] \\
    &= \Pr_{v \sim F}[v \ge \phi^{-1}_F(p) \mid v \ge p]    
    \end{align*}
    And it suffices for us to apply Lemma~\ref{lem:property-of-alpha-strongly-regular} for $\Pr_{v \sim G}[v \ge \phi^{-1}_G(0)] $.
\end{proof}

\paragraph{Revenue Maximization.}

For any truthful mechanism $\Pi=(\vb{x},\vb{p})$ for $n$ bidders with virtual values $\phi_1(v_1), \phi_2(v_2), \ldots, \phi_n(v_n)$, the virtual surplus of the allocation equals $\sum_{i \in [n]} \phi_i(v_i) x_i$. \cite{myerson81} shows that its expected revenue equals its expected virtual surplus.

\begin{theorem}[\cite{myerson81}]\label{thm:myerson}
    For any truthful mechanism $\Pi=(\vb{x},\vb{p})$, the expected revenue of $\Pi$ equals to its expected virtual surplus., i.e., 
    $$\E\left[\sum_{i\in [n]}p_i(\vb{v})\right] = \E\left[\sum_{i \in [n]} \phi_i(v_i) x_i(\vb{v}) \right].$$

    Furthermore, when $F_1,F_2,\ldots,F_n$ are regular distributions, the auction that gives optimal expected revenue (i.e., Myerson's optimal auction) will allocate the bidder $i^*\in [n]$ with highest virtual value $\phi_{i^*}(v_{i^*})\ge 0$ with a payment of $\phi_{i^*}^{-1}(\max\{0,\max_{j \ne i} \phi_j(v_j)\})$.
\end{theorem}

As a benchmark, let $\OptRev$ be the expected revenue obtained by using Myerson's optimal auction that maximizes virtual surplus for $n$ bidders with distribution $F_1,F_2,\ldots,F_n$ when there is no intermediary.
Similar to $\OptRev$, let $\APRev$ denote the expected revenue of the optimal anonymous posted-price mechanism (i.e., posting a price $p$ and selling the item to any bidder $i$ with $v_i \ge p$) without intermediary.
Previous works have showed that $\APRev$ can always recover a constant fraction of $\OptRev$.

\begin{theorem}[\cite{AHN+19}]
\label{thm:ratio-ap-to-opt-iid}
When $F_1,F_2,\ldots,F_n$ are identical regular distributions, $\APRev \ge (1-\frac{1}{e}) \OptRev$.
\end{theorem}

\begin{theorem}[\cite{JLQ+19}]
\label{thm:ratio-ap-to-opt-non-iid}
When $F_1,F_2,\ldots,F_n$ are non-identical regular distributions, $\APRev \ge \frac{1}{2.62} \OptRev$.
\end{theorem}

\section{Missing Proofs in \Cref{sec:seller-first}}
\label{appendix:missing-proof-seller-first}

\begin{proof}[Proof of \Cref{thm:seller-first-best-response}]
Recall that any $\Pi_I$ is an auction that sells the exclusive right to dictate the bids forwarded to $\Pi_S$ by \Cref{assumption:exclusive-right-dictate}.
Consider any intermediary's mechanism $\Pi_I=(\vb{\hat{b}},\vb{q})$.

Since we know the wining bidder $i^*$ in $\Pi_I$ (if exists) will always forward $\vb{\hat{b}^*}$ (or an equivalent bid vector) that minimizes the payment to obtain the item from the seller in $\Pi_S$, the seller will always allocate the item to bidder $i^*$ and charge her $p_0=\omega(\Pi_S)$.
By \Cref{thm:myerson}, the total expected revenue of $\Pi_I$ is equal to its expected virtual surplus:
\[\IRev_{\Pi_S,\Pi_I} = \mathbb{E}\left[\sum_{i =1}^n y_i (\phi_i(v_i) - p_0)\right].\]
Observe that such a quantity is maximized at $\vb{y^*}$ where $y_i^*=1$ if and only if $\phi_i(v_i) \ge \max\{p_0, \max_{j \ne i} \phi_j(v_j)\}$ (with ties broken arbitrarily), which exactly corresponds to the allocation rule of Myerson's optimal auction with a personalized reserve of $\phi^{-1}_i(p_0)-p_0$ (note that $b_i=v_i-p_0$ instead of $v_i$) when $F_1,F_2,\ldots,F_n$ are regular.
\end{proof}

\begin{proof}[Proof of \Cref{lem:example-BR}]
    Since there is only one bidder, the best response of the intermediary $\Pi_I^\BR(\Pi_S)$ from \Cref{thm:seller-first-best-response} is to run a Myerson optimal posted-price mechanism with some price $r \in [0,H]$ that maximizes the expected revenue $r(1-F_{\varepsilon,H}(r+p)) = \frac{r}{(r+p)^{1+\varepsilon}}$. We note that $r\geq 1-p$, since otherwise the intermediary can set $r=1-p$ and get higher revenue. Additionally, $r\leq H-p$ since otherwise $Pr_{v\sim F_{\varepsilon,H}}[v\geq p+r]=0$, yielding no revenue for the intermediary. By the first order condition for optimality, we get that the optimal price to set is $r^*(p) =\min\{p/\varepsilon, H-p\}$.
\end{proof}

\begin{proof}[Proof of \Cref{thm:max-ssm-opt-seller-leader}]
    When there is no intermediary, the optimal mechanism for the seller is a posted price mechanism with price $p\in [1,H]$ that maximizes the expected revenue $p(1-F_{\varepsilon,H}(p))=p\frac{1}{p^{1+\varepsilon}}=\frac{1}{p^\varepsilon}$.
    By setting $p=1$, one can obtain the optimal expected revenue $\OptRev=1$.

    We are left to analyze the seller's revenue in the presence of a best-responding intermediary.
    By Lemma~\ref{lem:example-BR}, when the seller uses a mechanism $\Pi_S$ with a minimum price $p$, the best response of the intermediary is to run the posted-price mechanism with price $r^*(p) =\min\{p/\varepsilon, H-p\}$. 
    Therefore, the seller should pick $p$ that maximizes her expected revenue 
    $$R_S(p) = \SRev_{\Pi_S,\Pi_I^\BR(\Pi_S)} = p(1-F_{\varepsilon,H}(r^*(p)+p))=\frac{p}{(r^*(p)+p)^{1+\varepsilon}}.$$
    
    Suppose the seller posts a price $p \ge \varepsilon H$, then we have $r^*(p)=H-p$ and the seller's expected revenue is 
    $$R_S(p)=\frac{p}{H^{1+\varepsilon}}.$$ 
    The seller can set a price $p=H$ to get a revenue of $H^{-\varepsilon}<(1+1/\varepsilon)^{-(1+\varepsilon)}$.
    On the other hand, when $p < \epsilon H$, we have $r^*(p)=p/\varepsilon$ and the seller's expected revenue is
    $$R_S(p) = \frac{p}{(p/\varepsilon+p)^{1+\varepsilon}}=\frac{1}{p^{\varepsilon}(1+1/\varepsilon)^{1+\varepsilon}}.$$ 
    The optimal price to set is $p=1$, yielding a revenue of $R_S(1)=(1+1/\varepsilon)^{-(1+\varepsilon)}$. 
    Combining these two cases, we conclude that the optimal price for the seller $p^*=1$ and $R_S(1) = (1+1/\varepsilon)^{-(1+\varepsilon)}$.
\end{proof}

\begin{proof}[Proof of \Cref{cor:lower-bound-strong-regularity}]
Let
$\varepsilon=\alpha/(1-\alpha)$.
Consider the distribution \(F_{\varepsilon,H}\) from
\Cref{thm:max-ssm-opt-seller-leader}, for \(H\) sufficiently large. Since $\frac{\varepsilon}{1+\varepsilon}=\alpha$, \(F_{\varepsilon,H}\) is \(\alpha\)-strongly regular.
Moreover, \Cref{thm:max-ssm-opt-seller-leader} shows that the seller's revenue
under any mechanism is at most
\[
(1+1/\varepsilon)^{-(1+\varepsilon)}\OptRev.
\]
Substituting \(\varepsilon=\alpha/(1-\alpha)\), we obtain
\[
(1+1/\varepsilon)^{-(1+\varepsilon)}
=
\left(\frac{1}{\alpha}\right)^{-1/(1-\alpha)}
=
\alpha^{1/(1-\alpha)}
=
c(\alpha).
\]
Thus for a bidder with valuation drawn from this $\alpha$-strongly regular distribution $F_{\alpha}=F_{\varepsilon,H}$, no seller's mechanism can guarantee more than $c(\alpha)*\OptRev$ revenue for the seller. 
\end{proof}

\begin{proof}[Proof of \Cref{cor:example-move-second-rev-opt}]
    We have shown the seller's optimal price $p^*=1$ in the proof above, and the intermediary's best response will be $r^*(1)=1/\epsilon$, yielding an expected revenue of $\frac{1/\epsilon}{(1+1/\epsilon)^{1+\epsilon}}$.
\end{proof}

\begin{proof}[Proof of \Cref{lem:alpha-strongly-regular-seller-revenue-to-ap}]
    Consider the seller running an anonymous posted-price auction with some fixed price $p_0>0$, denoted as $\Pi_i^{\AP(p_0)}=(\vb{x},\vb{p})$.
    By definition, the minimum achievable price for bidders is $\omega(\Pi_i^{\AP(p_0)})=p_0$ by submitting bids $(p_0,0,\ldots,0)$.
    Theorem~\ref{thm:seller-first-best-response} tells us the best response of the intermediary $\Pi_I^{\BR}(\Pi_i^{\AP(p_0)})$ is to run Myerson's optimal auction with personalized reserve $\phi^{-1}_i(p_0)-p_0$ for the bids $b_i = v_i-p_0$.
    
    We can now lower bound the seller's revenue as follows:
    \begin{align*}
        \SRev_{\Pi_S^{\AP(p_0)},\Pi_I^\BR(\Pi_S^{\AP(p_0)})}
        &= p_0 \cdot \Pr[\exists i, v_i-p_0\ge \phi_{i}^{-1}(p_0)-p_0]\\
        &= p_0 \cdot \sum_{i=1}^{n} (1-F_i(\phi^{-1}_i(p_0))) \prod_{j=1}^{i-1} F_j(\phi^{-1}_j(p_0)) \\
        &\ge c(\alpha) p_0 \cdot \sum_{i=1}^{n} (1-F_i(p_0)) \prod_{j=1}^{i-1} F_j(\phi^{-1}_j(p_0)) \tag{\Cref{cor:property-of-alpha-strongly-regular}}\\
        &\ge c(\alpha) p_0 \cdot \sum_{i=1}^{n} (1-F_i(p_0)) \prod_{j=1}^{i-1} F_j(p_0) \tag{$\phi_j(x)\le x$}\\
        &= c(\alpha) \cdot \Rev_{\Pi_S^{\AP(p_0)}},
    \end{align*}
    where $\Rev_{\Pi_S^{\AP(p_0)}}$ is the expected revenue for $\Pi_S^{\AP(p_0)}$ when there is no intermediary, and $\Rev_{\Pi_S^{\AP(p_0^*)}}=\APRev$ if we choose the optimal posted-price $p_0=p_0^*$.
    As a result, let $\Pi_S^{\AP^*}=\Pi_S^{\AP(p_0^*)}$ and we have $\SRev_{\Pi_S^{\AP^*},\Pi_I^\BR(\Pi_S^{\AP^*})} \ge c(\alpha) \cdot \APRev$.
\end{proof}

\begin{proof}[Proof of \Cref{lem:randomized-seller-revenue}]
    Fixing the seller's mechanism, the intermediary is a single-parameter mechanism design problem. By the revelation principle, we can assume the intermediary's auction is truthful. Let $x_{i,j}(v_i,\vec{v}_{-i})$ denote the probability of bidder $i$ getting allocated via the menu option $(\alpha_j,\beta_j)$ through the intermediary, when bidder values are $v_i, \vec{v}_{-i}$. By \Cref{thm:myerson}, we know the expected intermediary's revenue is equal to the expected virtual surplus:
    \begin{equation*}
      \E\left[\sum_{i=1}^n \sum_{j = 1}^m x_{i,j}(v_i) (\alpha_j \phi_i(v_i) - \beta_j)\right],
    \end{equation*}
    where $x_{i,j}(v_i)= E_{\vec{v}_{-i}}[x_{i,j}(v_i,v_{-i})]$. Since the intermediary only selects one bidder to forward its message by \Cref{assumption:exclusive-right-dictate} and each bidder wants to pick the best option on the menu, we have $\sum_{i=1}^{n}\sum_{j=1}^{m}x_{ij}(v_i,\vec{v}_{-i})\leq 1, \forall v,\vec{v}_{-i}$. Thus the expected virtual surplus can be upper bounded as
    $$\E[ \max(0, \max_{i \in [n], j \in [m]} (\alpha_j \phi_i(v_i) - \beta_j))] = \E[ \max(0,\max_{j \in [m]} (\alpha_j \max_{i \in [n]} \phi_i(v_i) - \beta_j))],$$
    which can be achieved if the intermediary allocates the item to the bidder $i^*$ with the maximum virtual value via the option $(\alpha_{j^*},\beta_{j^*})$ that gives a maximum non-negative virtual surplus, i.e.,
    \[\begin{cases}
     i^*=\argmax_{i \in [n]} \phi_i(v_i)\\
     j^*=\argmax_{j \in [n]} (\alpha_j \phi_{i^*}(v_{i^*}) - \beta_j) 
    \end{cases}
    \text{whenever $\alpha_{j^*} \phi_{i^*}(v_{i^*}) - \beta_{j^*} \ge 0$}.
    \]
    Note that this allocation is monotone in each bidder's reported value and is thus implementable by a truthful direct mechanism.
    Therefore, when the the intermediary best responds, the seller's expected revenue is
    \[\E[\beta_j^* \I[\alpha_{j^*} \phi_{i^*}(v_{i^*}) - \beta_{j^*} \ge 0 ]]. \qedhere\]
\end{proof}

\begin{proof}[Proof of \Cref{thm:randomization-does-not-help-seller}]
    By \Cref{lem:randomized-seller-revenue}, observe that the seller's revenue using a menu $\{(\alpha_j,\beta_j)\}_{j \in [m]}$ is exactly equal to the revenue of selling an item to a single bidder with a value $\overline{v}=\max_{i \in [n]} \phi_i(v_i)$ using the same menu.
    However, since we know that Myerson's optimal auction in the latter setting is simply a deterministic posted-price mechanism by \Cref{thm:myerson}, we conclude that the randomization cannot give better revenue for the seller.
\end{proof}

\section{Missing Proofs in \Cref{sec:bp-first}}
\label{appendix:missing-proof-bp-first}

\begin{proof}[Proof of \Cref{lem:maximum-mhr-remains-mhr}]
It is easy to see the CDF $F_{\max}=F^n$. Thus we will use the distribution $F^n$ from now on. Recall that $F$ being MHR means its hazard rate function $h_F(x) = \frac{f(x)}{1-F(x)}$ is non-decreasing, where $f(x)=F'(x)$ is the PDF of $F$.
Similarily, the harzard rate function for $F^n$ is
$$h_{F^n}(y) = \frac{(F^n(y))'}{1-F^n(y)} = \frac{n F^{n-1}(y) f(y)}{1-F^n(y)} = \frac{n F^{n-1}(y) h_F(y) (1-F(y))}{1-F^n(y)}= n h_F(y) \frac{F^{n-1}(y) }{\sum_{k=0}^{n-1} F^k(y)}.$$
Note that $n h_F(y)$ and $F(y)$ is non-decreasing, and $g(t)=\frac{t^{n-1}}{\sum_{k=0}^{n-1} t^k}$ is also non-decreasing for $t \in (0,1)$ by observing that $\frac{1}{g(t)}=\sum_{k=0}^{n-1} t^{-k}$ is positive and non-increasing over $(0,1)$.
Therefore, $h_{F^n}(y)$ is non-decreasing and thus $F^n$ is also MHR.
\end{proof}

\begin{proof}[Proof of \Cref{thm:bp-first-mhr}]
    Let $F_i=F$ be the valuation distribution for all $i\in [n]$, and let $\phi_{F^n}$ denote the virtual value function for the distribution of the maximum of $n$ random variables with distribution $F$. 
    Since $F$ is MHR, the distribution $F^n$ of maximum of $n$ i.i.d. MHR random variables  is also MHR by \Cref{lem:maximum-mhr-remains-mhr}.
    Suppose the intermediary runs a posted price mechanism with price $r$, denoted $\Pi_I=\Pi_I^{\AP(r)}$. Then the seller's best response is to post the following fixed price $p(r)$ by \Cref{thm:active-bp-best-response-new-model}:
    \[p(r)=\argmax_p\{p\cdot B_{\Pi_{I},p})\}\]
    where $B_{\Pi_I,p} = 1-F^n(p+r)$ is the probability that some bidder wants to buy the item when they need to pay the intermediary a price of $r$ and the seller a price of $p$.
    The optimal price the seller should set is thus $p(r)=\phi_{F^n}^{-1}(r)-r$ using the first-order optimality condition.
    Now, by setting $r^*=\argmax_r \{r(1-F^n(r))\} = \phi^{-1}_{F^n}(0)$, the intermediary can obtain a maximum revenue of
    \begin{align*}
        \IRev_{\Pi^{\BR}_{p}(\Pi_r), \Pi_r} &= r^*(1-F^n(r^*+p(r^*))\\
        &=r^{*}(1-F^n(r^{*}+\phi_{F^n}^{-1}(r^{*})-r^{*}))\\
        &=r^{*}(1-F^n(\phi_{F^n}^{-1}(r^{*})))\\
        &\geq \frac{1}{e} r^{*}(1-F^n(r^{*})) \tag{by \Cref{cor:property-of-alpha-strongly-regular}}\\
        &= \frac{1}{e} \max_{r'}r'
        (1-F^n(r'))= \frac{1}{e} \APRev \tag{since $r^{*}=\arg\max_r \{r(1-F^n(r))\} = \phi^{-1}_{F^n}(0)$}\\
        &\geq \frac{1}{e} (1-\frac{1}{e}) \OPT \tag{by \Cref{thm:ratio-ap-to-opt-iid} since MHR distributions are regular}
    \end{align*}
    This completes our proof.
\end{proof}

\section{Missing Proofs in \Cref{sec:simultaneous}}
\label{appendix:missing-proof-simultaneous}

\begin{proof}[Proof of \Cref{thm:sim-unbdd}]
    Let the bidder have valuation distributed according to $F_{\epsilon,H}$ for some $\epsilon>0$. We know that $F_{\epsilon,H}$ is $\alpha$-strongly regular with $\alpha = \frac{\epsilon}{1+\epsilon}$. Since we are in the single-bidder case, the best mechanism for both the seller and the intermediary is a posted-price mechanism. Let $(p,r)$ denote the prices the seller and the intermediary set, respectively. For $(p,r)$ to be a Nash equilibrium, we require both posted price mechanisms to be a best response to each other. By Lemma~\ref{lem:example-BR} and the symmetry in the simultaneous model, we have $r =\min\{p/\epsilon, H-p\}$ and $p =\min\{r/\epsilon, H-r\}$. This means the only Nash equilibria are of the form $(p=x, r=H-x)$ where $\frac{H}{1+1/\epsilon}\leq x\leq \frac{H}{1+\epsilon}$. To see why, consider the following two cases. First suppose $p+r\neq H$, then $p=r/\epsilon, r=p/\epsilon$, which implies $p=r=0$. However, this is clearly not a Nash equilibrium as either player can set a price of 1 instead to get the maximum revenue 1. Now suppose $p+r=H$, then $r=H-p\leq p/\epsilon$ and $p=H-r\leq r/\epsilon$, implying $\frac{H}{1+1/\epsilon}\leq x\leq \frac{H}{1+\epsilon}$. 
    
    Let $\Pi_S=\Pi_S^{AP(x)}$ and $\Pi_I=\Pi_I^{AP(H-x)}$. Then $\max \{\SRev_{\Pi_S, \Pi_I}, \IRev_{\Pi_S, \Pi_I}\}\leq \frac{\max\{x,H-x\}}{H^{1+\epsilon}}\leq \frac{1}{H^\epsilon}$. However, when either the seller or the intermediary goes first, they can get a revenue of $\frac{1}{(1+1/\epsilon)^{1+\epsilon}}$ by Theorem~\ref{thm:max-ssm-opt-seller-leader}. We also know that when the seller or the intermediary goes second, they can get a higher revenue of $\frac{1/\epsilon}{(1+1/\epsilon)^{1+\epsilon}}$ by Corollary~\ref{cor:example-move-second-rev-opt}, so the gap $\frac{\max \{\SRev_{\Pi_S, \Pi_I}, \IRev_{\Pi_S, \Pi_I}\}}{\min \{\SRev,\IRev\}}\leq\frac{(1+1/\epsilon)^{1+\epsilon}}{H^\epsilon}$, which goes to 0 as $H\to \infty$.
\end{proof}